\documentclass[11pt,twoside]{article}

\usepackage[ansinew]{inputenc} 
\usepackage{amsmath,amsfonts,amssymb,amsthm}
\usepackage{pstricks,pst-node,pst-coil,pst-plot,pstricks-add}
\usepackage{graphics,geometry,epsfig}
\usepackage{bbm}

\usepackage{mathrsfs} 

\newcommand{\field}[1]{\mathbb{#1}}

\newcommand{\R}{\field{R}}
\newcommand{\C}{\field{C}}

\newcommand{\HH}{\mathscr H}

\newcommand{\EE}{\mathcal E}
\newcommand{\FF}{\mathcal F}

\newcommand{\Hres}{H_{\alpha} }    
\newcommand{\Hrescut}{H_{\alpha,\Lambda} }    
\newcommand{\Heff}{H_{\ph_0} }   

\newcommand{\eps}{\varepsilon}
\newcommand{\ph}{\varphi}

\newcommand{\ran}{\mathrm{Ran}}

\newcommand{\Nph}{N_\mathrm{ph}}

\newcommand{\restricted}{|\grave{}\,}

\newcommand{\expect}[1]{\mbox{$\langle #1 \rangle $}}         
\newcommand{\sprod}[2]{\mbox{$\left\langle #1,#2 \right\rangle$}}        

\newcommand{\Ima}{\operatorname{Im}}
\newcommand{\Rea}{\operatorname{Re}}

\newtheorem{theorem}{Theorem}[section]
\newtheorem{lemma}[theorem]{Lemma}

\newtheorem{prop}[theorem]{Proposition}

\theoremstyle{plain}

\title{\textbf{On the dynamics of polarons in the strong-coupling limit}}

\author{Marcel Griesemer\\  
\small Fachbereich Mathematik, Universit\"at Stuttgart, D-70569 Stuttgart, Germany\\
\small marcel.griesemer@mathematik.uni-stuttgart.de}  

\date{}

\begin{document}
\maketitle
\begin{abstract}
The polaron model of H. Fr\"ohlich describes an electron coupled to the quantized longitudinal optical modes of a polar crystal. In the strong-coupling limit one expects that the phonon modes may be treated classically, which leads to a coupled Schr\"odinger-Poisson system with memory. For the effective dynamics of the electron this amounts to a nonlinear and non-local Schr\"odinger equation. We use the Dirac-Frenkel variational principle to derive the Schr\"odinger-Poisson system from the Fr\"ohlich model and we  present new results on the accuracy of their solutions for describing the motion of Fr\"ohlich polarons in the strong-coupling limit. Our main result extends to $N$-polaron systems.
\end{abstract}

\section{Introduction}

In the Fr\"ohlich model of large polarons the energy of an electron in a ionic crystal is described a Hamiltonian $H_F$ that may be formally written as 
\begin{equation}\label{pure-ham}
      H_F = -\Delta + \alpha^{-2}N + \alpha^{-1}W,
\end{equation}
where $\Delta$ the Laplacian in $L^2(\R^d)$ and $N$ is the number operator on symmetric Fock space, $\FF$, over  $L^2(\R^d)$. The operator $W$ is linear in creation and annihilation operators and accounts for the interaction of electron an phonons. The explicit form of this interaction depends on the space dimension $d\in\{2,3\}$ and is given in Section~\ref{DF-polaron}. The parameter $\alpha>0$ is a dimensionless coupling constant and large $\alpha$ means \emph{strong} coupling.  See Section~\ref{DF-accuracy} for an explanation of our units in \eqref{pure-ham}. 

We are interested in the dynamics generated by $H_F$ in the case of large $\alpha$ and our work is inspired by the following well-known result on the ground-state energy $E_F=\inf\sigma(H_F)$: Let $E_P$ be the minimum of $\sprod{\psi}{H_F\psi}$ with respect to all product states $\psi=\ph\otimes \eta$ of normalized $\ph\in L^2$ and $\eta\in\FF$. Then 
$E_P$, which is known as the Pekar energy, is an upper bound on $E_F$ and $E_P\to E_F$ in the limit $\alpha\to\infty$ \cite{DV1983}. More precisely  \cite{LiebThomas1997},
\begin{equation}\label{in:LT}
      E_F \leq E_P \leq E_F + O(\alpha^{-1/5})\qquad (\alpha\to\infty).
\end{equation}
The Pekar energy $E_P$, by its very definition, is the minimum of the Pekar functional $\ph\mapsto \min_{\|\eta\|=1}\sprod{\ph\otimes\eta}{H_F(\ph\otimes \eta)}$ on the sphere $\|\ph\|=1$ of $L^2$. Any minimizer solves the corresponding Euler-Lagrange equation
\begin{equation}\label{in:Pekar}
      (-\Delta + V)\ph = \lambda \ph,\qquad V = - |\cdot|^{-1}*|\ph|^2,
\end{equation}
where $V$,  for given $\ph$, is the potential generated by the (coherent) state minimizing $\sprod{\ph\otimes \eta}{H_F\ph\otimes\eta}$. In physical terms it is the potential due to the deformation of the ionic lattice by the electron. It is well-known that the Pekar functional has a minimizer and hence that \eqref{in:Pekar} has a solution.

The question arises whether the equations~\eqref{in:Pekar} have a generalization to a time-dependent setting, where they approximately describe the motion of a polaron in the case $\alpha \gg 1$. The answer to this question from the literature is the system 
\begin{equation}\label{in:LP}
      (-\Delta + V)\ph = i\partial_t \ph,\qquad (\alpha^4 \partial_t^2 + 1)V = - |\cdot|^{-1}*|\ph|^2,
\end{equation}
which apparently was first written down, in an equivalent form, by  Landau and Pekar \cite{LP1948, DevAle2009}. The mathematical analysis of this system was initiated in  \cite{BNRS2000} and continued in
\cite{FG2015,GSS2016}. Note that \eqref{in:LP} reduces to \eqref{in:Pekar} in the stationary case where $V$ is independent of time and  $\ph_t = e^{-i\lambda t}\ph_0$.

The accuracy of the effective evolution defined by equations that are closely related to \eqref{in:LP} was  recently analyzed in two papers by Frank with Schlein and Gang, respectively \cite{FrankSchlein2014, FG2015}. In \cite{FrankSchlein2014} the effective dynamics is linear and it affects the electron only, the phonons being frozen. The analysis in \cite{FG2015} is much refined in comparison to \cite{FrankSchlein2014}. The results concern the effective evolution of product states defined by the Landau-Pekar equations and the limitations of such an approximation are discussed as well. The error bounds in \cite{FrankSchlein2014, FG2015} are good for $|t| = o(\alpha)$\footnote{The Gronwall argument in \cite{FrankSchlein2014} is easily modified to give $|t| = o(\alpha)$, rather than  $|t| = o(\log \alpha)$.}. 

In this paper we derive the system \eqref{in:LP} from the Schr\"odinger equation $i\partial_t \psi = H_F\psi$ with the help of the Dirac-Frenkel time-dependent variational principle in connection with the manifold of all product states $\ph\otimes\eta\in L^2\otimes \FF$ \cite{Lubich}. This methods provides us with a coupled system of Schr\"odinger equations for $\ph$ and $\eta$, we call them Dirac-Frenkel equations, where the one for $\eta$ can be solved explicitly and then yields a potential $V$ satisfying \eqref{in:LP}. Likewise, the Landau-Pekar equations considered in \cite{FG2015} follow from the Dirac-Frenkel equations for the polaron. We stress that the Dirac-Frenkel principle is \emph{not} based on unjustified assumptions (like propagation of chaos) and that it automatically gives the correct phase for the approximating evolution. Moreover, the variational principle  of Dirac and Frenkel provides one with conservation laws, a geometric picture, and analytic relations that are useful for the comparison with the true quantum evolution. There is a particular solution to the Dirac-Frenkel equations of the form $t\mapsto e^{-iE_Pt}(\ph_0\otimes\eta_0)$, where $\ph_0$ is a minimizer of the Pekar functional and $\eta_0$ is the coherent state associated with $\ph_0$. It corresponds to the solution of \eqref{in:LP}, where $V$ is time-independent and thus \eqref{in:LP} reduces to \eqref{in:Pekar}. Armed with the method of Dirac and Frenkel and with some inspiration from \cite{FrankSchlein2014} we  show in this paper that
\begin{equation}\label{main-result}
        \|e^{-iHt}(\ph_0\otimes\eta_0) - e^{-iE_Pt}(\ph_0\otimes\eta_0)\|^2 \leq C \frac{|t|}{\alpha^2},
\end{equation}
for all $t\in\R$ and $\alpha\geq 1$. This expands by a power of two the time scale for which a solution to \eqref{in:LP} has previously been found in good agreement with the true quantum evolution. In addition it shows that the state $\ph_0\otimes \eta_0$, whose energy $E_P = \sprod{\ph_0\otimes\eta_0}{H(\ph_0\otimes\eta_0)}$, by \eqref{in:LT}, is close to $E_F$, is very stable for large $\alpha$, which is consistent with the fact that the polaron mass grows to infinity as $\alpha\to \infty$ \cite{Spohn1987}.
As a further illustration of our methods we compare the Dirac-Frenkel evolution $u_t = a(t)\ph_t\otimes \eta_t$ to the true polaron evolution for a fairly general class of initial states $\ph_0\otimes \eta_0$, with $\eta_0$ being a coherent state. We find that 
\begin{equation}\label{main-two}
        \|e^{-iHt}(\ph_0\otimes\eta_0) - u_t\|^2  = O(|t|/\alpha),
\end{equation}
for $\alpha\geq 1$ and all $t\in\R$. The evolution $t\mapsto u_t$ agrees with the evolution of states of $\ph_0\otimes \eta_0$ considered in \cite{FG2015}. The very involved analysis in \cite{FG2015} yields the better bound $O(|t|/\alpha)^2$ for the left-hand side of \eqref{main-two}. Our main point here is that \eqref{main-two} follows from a very short argument exploiting general properties of the Dirac-Frenkel variational principle and the fact that $\eta_0$ is a coherent state.

In Section \ref{DF-abstract} the time-dependent variational principle of Dirac and Frenkel is described in its abstract form, and in Section~\ref{DF-polaron} it is applied to the polaron.
In particular, the system \eqref{in:LP} is derived in Section~\ref{DF-polaron}. Our main result, \eqref{main-result}, is established in Section~\ref{DF-accuracy} for polarons in two and tree space dimensions and in Section~\ref{sec:N-polaron} it is extended to multipolaron systems in three space dimensions. An important technical tool in our work is Lemma~\ref{lm:FSch}, which we learned from \cite{FrankSchlein2014}. This lemma is very useful for controlling the mild ultraviolet problem of \eqref{pure-ham}. It has also been used effectively in \cite{GriWue2016} for describing the domain of self-adjointness of the Fr\"ohlich Hamiltonian. Estimate \eqref{main-two} is proved in Section~\ref{DF-general}.

In this paper the Fourier transform $\hat\rho$ of a function $\rho\in L^1(\R^d)$ is defined by
$$
         \hat\rho(k):= \int e^{-ikx}\rho(x)\,dx,
$$
so that $\widehat{f*g} = \hat{f}\cdot \hat{g}$ and $\sprod{f}{g} = (2\pi)^{-d}\langle\hat{f},\hat{g}\rangle$. In $d=3$ it follows that $|x|^{-1}$ is sent to $4\pi/|k|^2$ under Fourier transform.

\section{The Dirac-Frenkel Approximation}
\label{DF-abstract}

In this section we present the Dirac-Frenkel approximation in its abstract form and for the manifolds of our interest. For a more elaborate discussion, the history and more applications we refer to the ``blue book" of Lubich \cite{Lubich}.

Let $\HH$ be some complex Hilbert space and let $M\subset \HH$ be a
subset of $\HH$ which is a manifold in the following sense: we assume
that the set of velocities 
$$
      \dot{u} = \frac{d}{dt}u_t\Big|_{t=0}
$$
of differentiable curves $t\mapsto u_t\in M$ with $u_{t=0}=u$ is a
closed subspace of $\HH$. This subspace is denoted $T_{u}M$ and called tangent space 
of $M$ at the point $u\in M$. Its closedness ensures the existence of an orthogonal
projection $P(u):\HH \to T_{u}M$ onto $T_{u}M$. Now let $H:D\subset \HH\to \HH$ be some self-adjoint operator in $\HH$ and
let $u_0\in M$. Then any solution $u$ of the Dirac-Frenkel equation
\begin{equation}\label{DF}
         \dot{u} = -iP(u) Hu\qquad u|_{t=0}=u_0
\end{equation}
may be considered an approximation to the Schr\"odinger
evolution $t\mapsto e^{-iHt}u_0$. In applications we expect that the
solution to \eqref{DF} is unique and we are interested in the accuracy
of the Dirac-Frenkel approximation $u$.  A simple bound is obtained as
follows: if $t\mapsto u_t$ is a solution to \eqref{DF} then 
$$
   \frac{d}{dt}e^{iHt}u_t = ie^{iHt}P(u)^{\perp}Hu
$$
which integrates to 
$$
       e^{iHt}u_t-u_0 = i\int_0^t e^{iHs}P(u_s)^{\perp}Hu_s\, ds.
$$
It follows that 
$$
     \|u_t - e^{-iHt}u_0\| \leq \int_0^t \|P(u)^{\perp}Hu\|\, ds.
$$
Note that the energy $\sprod{u}{Hu}$ is conserved along solutions to
\eqref{DF} because
$$
    \frac{d}{dt}\sprod{u}{Hu} = 2\Rea\sprod{Hu}{\dot u} =
    2\Rea(-i)\sprod{Hu}{P(u)Hu} =0. 
$$
By a similar computation the norm $\|u\|$ is conserved provided that
$u\in T_{u}M$ for all $u\in M$, which will be the case for our choice of $M$. 

In this paper we are interested in Hilbert spaces of the form $\HH = \HH_1\otimes\HH_2$ and the submanifold
$$
     M:= \{u\in \HH \mid u=a\ph\otimes \eta \neq 0,\
     \text{where}\ a\in\C,\ \ph\in\HH_1,\eta\in \HH_2\}.
$$
The tangent space $T_{u}M$ of $M$ at the point $u=a\ph\otimes \eta$
consists of all vectors of the form
\begin{equation}\label{dot-u}
     \dot{u} = \dot{a} \ph\otimes \eta + a\dot{\ph}\otimes \eta + a\ph\otimes \dot{\eta}
\end{equation}
with arbitrary $\dot{a} \in\C$, $\dot{\ph}\in \HH_1$ and $\dot{\eta}\in
\HH_2$. These representations of $u$ and $\dot{u}$ are not unique and the
liberty in the choice of the parameters $a,\ph,\eta,
\dot{a},\dot{\ph},\dot{\eta}$ is called gauge freedom. Let us fix the gauge by imposing 
the gauge conditions 
\begin{equation}\label{normal}
     \|\ph\|= 1,\qquad\|\eta\|=1
\end{equation}
and 
\begin{equation}\label{gauge}
 \sprod{\dot{\ph}}{\ph} = 0,\qquad \sprod{\dot{\eta}}{\eta} = 0.
\end{equation}
The second condition is achieved by the gauge transform
$\dot{a}\mapsto \dot{a}+(\lambda+\mu)a$, $\dot{\ph}\mapsto \dot{\ph}-\lambda\ph$, and
$\dot{\eta}\mapsto\dot{\eta}-\mu\eta$ with suitable $\lambda,\mu\in
\C$. These gauge conditions can be imposed on on any differentiable
curve $u=a(\ph\otimes\eta)$ in $M$. Indeed, the normalization of $\ph$
and $\eta$ preserves differentiability and once \eqref{normal} is
satisfied, the gauge transform
\begin{align*} 
    \ph &\mapsto \ph\exp\left(-\int \sprod{\ph}{\dot\ph}dt\right)\\
    \eta &\mapsto \eta\exp\left(-\int \sprod{\eta}{\dot\eta}dt\right)\\
    a &\mapsto a \exp\left(\int \sprod{\eta}{\dot\eta}+ \sprod{\ph}{\dot\ph}\,dt \right)\
\end{align*}
preserves the norms and leads to \eqref{gauge}. 

The mere possibility to represent of tangent vectors  \eqref{dot-u} in the gauge \eqref{gauge}
shows that $P(u)$ at the point $u=a\ph\otimes \eta$ can be written as a
sum of three orthogonal and mutually orthogonal projections
\begin{equation}\label{Pu}
     P(u) = P_{\ph\otimes \eta } + P_\ph^{\perp}\otimes P_\eta +
     P_\ph\otimes P_\eta^{\perp},
\end{equation}
where $P_\ph$ denotes the orthogonal projection onto the
one-dimensional space spanned by the vector $\ph$.  Other useful formulas for $P(u)$ are 
\begin{align}
     P(u) &= P_\ph\otimes 1 + 1\otimes P_\eta - P_{\ph\otimes \eta } \label{Pu2}\\
     P(u)^{\perp}  &=  P_\ph^{\perp}\otimes P_\eta^{\perp}.\nonumber
\end{align}

Suppose $u=a\ph\otimes \eta$ is a curve in $M$ written in the gauge
\eqref{gauge}, $a$, $\ph$ and $\eta$ being differentiable curves in $\C$, $\HH_1$ and $\HH_2$, respctively. Then the Dirac-Frenkel equation \eqref{DF}
splits into three equations associated with the three mutually orthogonal subspaces of
$T_uM$ defined by \eqref{Pu}. These equations read 
\begin{align}
    i\dot{a} &= \expect{H}_{\ph\otimes\eta}a \label{dot-a}\\
    i\dot{\ph} &= (H_\eta -  \expect{H}_{\ph\otimes\eta})\ph \label{dot-ph}\\
     i\dot{\eta} &= (H_\ph -  \expect{H}_{\ph\otimes\eta})\eta \label{dot-eta}
\end{align}
where $ \expect{H}_{\ph\otimes\eta} =
\sprod{\ph\otimes\eta}{H(\ph\otimes\eta)} / \sprod{\ph\otimes\eta}{\ph\otimes\eta}$ and $H_\eta$, $H_\ph$ are
operators defined by 
\begin{align*}
      (1\otimes P_{\eta})H(\ph\otimes\eta) &= H_\eta\ph\otimes \eta\\
     (P_\ph\otimes 1) H (\ph\otimes\eta) &= \ph\otimes H_\ph\eta,
\end{align*}
respectively. With some abuse of notation we may write
$H_\eta=\sprod{\eta}{H\eta}/\sprod{\eta}{\eta}$ and $H_\ph=\sprod{\ph}{H\ph}/\sprod{\ph}{\ph}$ with
inner products are taken in $\HH_2$ and $\HH_1$ only. 

In the case where $H=H_0+W$ with 
$$
    H=H_0+W\qquad H_0 = H_1\otimes 1 + 1\otimes H_2
$$
it follows that $H_{0,\eta} - \expect{H_0}_{\ph\otimes \eta} = H_1- \expect{H_1}_\ph$ and $H_{0,\ph} - \expect{H_0}_{\ph\otimes \eta} = H_2- \expect{H_2}_\eta$. 
The Equations \eqref{dot-a}-\eqref{dot-eta} therefore become
\begin{align}
    i\dot{a} &= (\expect{H_1}_\ph+ \expect{H_2}_\eta +\expect{W}_{\ph\otimes\eta})a \label{dot-a2}\\
    i\dot{\ph} &= (H_1+ W_\eta - \expect{H_1}_\ph - \expect{W}_{\ph\otimes\eta})\ph \label{dot-ph2}\\
     i\dot{\eta} &= (H_2+ W_\ph -  \expect{H_2}_\eta -\expect{W}_{\ph\otimes\eta})\eta. \label{dot-eta2}
\end{align}
We now collect the scalar contributions to \eqref{dot-ph2} and \eqref{dot-eta2} in \eqref{dot-a2}. This is achieved
by a change of phase in $\ph,\eta$ that violates the gauge
condition \eqref{gauge} but preserves the norms. We conclude that
any solution $u=a(\ph\otimes\eta)$ to \eqref{DF} may be gauged in such
a way that $a,\ph,\eta$ solve the system
\begin{align}
    i\dot{a} &= -\expect{W}_{\ph\otimes\eta} a \label{dot-a3}\\
    i\dot{\ph} &= (H_1+ W_\eta)\ph \label{dot-ph3}\\
     i\dot{\eta} &= (H_2+ W_\ph)\eta \label{dot-eta3},
\end{align}
where the first equation immediately integrates to 
\begin{equation}\label{a-of-t}
     a(t) = a(0)\exp \left(i\int_0^t \expect{W}_{\ph\otimes \eta}\,ds\right).
\end{equation}

Let us prove the converse, i.e., that \eqref{dot-ph3} through \eqref{a-of-t} are sufficient for \eqref{DF}. Indeed, from $P(u)Hu=H_0u+P(u)Wu$ and the representation \eqref{Pu2} for $P(u)$ we see that \eqref{DF} becomes 
\begin{equation}\label{eq:DF-ham}
    i\dot u = \big[(H_1+ W_\eta) \otimes 1 + 1\otimes (H_2 + W_\ph) -\expect{W}_{\ph\otimes\eta}\big] u
\end{equation}
which is obviously satisfied if $a$, $\ph$, and $\eta$ solve \eqref{dot-a3} through \eqref{dot-eta3}. 

Note that the norms of $\ph$ and $\eta$ are conserved for solutions of \eqref{dot-ph3} and \eqref{dot-eta3} by the (assumed) symmetry of the operators $H_1$, $H_2$, and $W$. Hence, if $\|\ph(0)\|=1=\|\eta(0)\|$ then $\|\ph\|=1=\|\eta\|$ for all times and 
$\expect{W}_{\ph\otimes \eta} = \sprod{\ph\otimes\eta}{W(\ph\otimes \eta)}$.



\section{The Dirac-Frenkel approximation for the polaron}
\label{DF-polaron}

In this section we derive various equivalent forms of the Dirac-Frenkel equations for the polaron and the manifold associated with the factorization of Hilbert space into particle and phonon spaces. The Landau-Pekar equations considered in \cite{FG2015} are one possible form of the Dirac-Frenkel equations. 

Let $\HH=L^2(\R^d)\otimes\FF$ where $\FF$ denotes the symmetric Fock space over $L^2(\R^d)$. Let $N$ denote the number operator in $\FF$, let $a^{*}(f)$ and $a(f)$ be the usual creation and annihilation operators associated with a function $f\in L^2(\R^d,dk)$, and let $\phi(f) = a(f)+a^{*}(f)$. We define (unitary) Weyl operators $W(f)$ by
$$
               W(f) = e^{-i\phi(if)} = \exp(a^{*}(f)-a(f)).
$$
It is useful and convenient to extend the notions of creation and annihilation operators to include operators $a^{*}(G)$ and $a(G)$ in $\HH$ where the map $k\mapsto G(k)$ takes values in the bounded operators on the particle space $L^2(\R^d)$. With these preparations we can defined the Fr\"ohlich Hamiltonian $H_{F}$ as the norm-resolvent limit as 
$\Lambda\to\infty$ of 
\begin{equation}\label{Fro-Ham}
     H_{F,\Lambda} = -\Delta + \omega_0 N + \sqrt{\alpha}\phi(G_\Lambda)
\end{equation}
where $\omega_0,\alpha>0$,
$$
     \phi(G_\Lambda) = \int \big(G_\Lambda(k)^{*}\otimes a(k) + G_\Lambda(k)\otimes a^{*}(k)\big)dk
$$
and $G_\Lambda(k)$ denotes the bounded operator $L^2(\R^d)$ given by multiplication with the function $G_{\Lambda,x}(k)= \chi_\Lambda(k)v(k)e^{-ikx}$ \cite{GriWue2016}. We choose $v\in L^2_{\rm loc}(\R^d)$ to be the positive function defined by 
$$
       \frac{1}{2|x|} = \int e^{ikx} |v(k)|^2\, dk.    
$$
This means that $v(k) \sim |k|^{-(d-1)/2}$, which is the usual choice for the form factor in two and three dimensions \cite{Devreese}. In particular we get the familiar
$$
       v(k) =\frac{1}{2\pi |k|},\qquad \text{for}\  d=3.
$$

While $H_{F,\Lambda}$ is an operator sum of the general form $H_0+W$ of Hamiltonians discussed in the previous section, this is not true for the resolvent limit $H_F$ \cite{GriWue2016}. We ignore this problem in the present section and we derive (well-posed) systems of effective equations for the evolution on the manifold of product stated introduced in Section~\ref{DF-abstract}. This procedure will be justified by the results of the next section.

In this spirit we formally set  $W=\sqrt{\alpha}\phi(G)$, where $G=G_{\Lambda=\infty}$, and we find
\begin{align*}
     W_\eta &=  \sqrt{\alpha}\int \big(G(k)^{*}\sprod{\eta}{a(k)\eta} + \text{h.c.}\big)dk\\
    W_\ph &=  \sqrt{\alpha}\int \big(\sprod{\ph}{G(k)^{*}\ph} a(k) +  \text{h.c.}\big)dk.
\end{align*}
These operators are of the form $W_\eta = V_\eta$ and
$W_\ph=\sqrt{\alpha}\phi(f_\ph)$ with functions 
\begin{align}\label{def-V}
    V_\eta(x) &= \sqrt{\alpha}\int \left(\overline{v(k)}e^{ikx}\sprod{\eta}{a(k)\eta}
      +  \text{h.c.} \right)dk\\
   f_{\ph}(k) &= \sprod{\ph}{G(k)\ph} = v(k) \int e^{-ikx}|\ph(x)|^2\,dx\label{def-f}
\end{align}
Hence the Dirac-Frenkel equations \eqref{dot-ph3} and \eqref{dot-eta3} become
\begin{align}
    i\dot{\ph} &= (-\Delta + V_\eta)\ph,\label{DF-phi}\\ 
   i\dot{\eta} &= (\omega_0 N + \sqrt{\alpha}\phi(f_t))\eta,\label{DF-eta} 
\end{align}
and 
\begin{equation}\label{DF-a}
      a(t) = \exp \left( i\sqrt{\alpha}\int_0^t \sprod{\ph_s\otimes \eta_s}{\phi(G)(\ph_s\otimes \eta_s)}\,ds\right)
\end{equation}
provided that $\|\ph(0)\|=1=\|\eta(0)\|$ and $a(0)=1$. The definitions \eqref{def-V} and \eqref{def-f} with $\sqrt{\alpha}$ included in $V_\eta$ but not in $f_{\ph}$ will turn out convenient later where both $V$ and $f$ will become independent of $\alpha$.

We are now going to eliminate $\eta$ from \eqref{DF-phi} and \eqref{DF-eta} and we shall replace \eqref{DF-eta} by an equation for $V_\eta$.
The equation \eqref{DF-eta} can be solved for $\eta$ in spite of the fact
that it is non-autonomous.  To this end we pass to the interaction picture. That is, we define 
$\eta_{\rm int}(t) = e^{i\omega_{0}Nt} \eta(t)$ for which Equation \eqref{DF-eta} becomes 
$$
     i\dot\eta_{\rm int} = \sqrt{\alpha}\phi(e^{i\omega_0t}f_\ph))\eta_{\rm int}.
$$
Since the commutator 
$$
       [\phi(e^{i\omega_0t}f_t)), \phi(e^{i\omega_0s}f_s))] = 2i\Ima\sprod{f_t}{e^{-i\omega_0(t-s)}f_s}
$$
is a number, there is a formula for the propagator, which for $\eta(t)=e^{-i\omega_0Nt}\eta_{\rm int}(t)$ implies 
\begin{align}\label{eta-sol}
     \eta(t) &= e^{-iF(t)} e^{-i\omega_{0}Nt} W(J_t)\eta_0\\
  \text{where}\quad   J_t(k) &= -i\sqrt{\alpha}\int_0^t e^{i\omega_{0}s}f_s(k)\, ds,\\
    F(t) &= -\alpha\int_0^{t}ds_1 \int_0^{s_2} ds_2 \Ima\sprod{f_{s_1}}{e^{-i\omega_0(s_1-s_2)}f_{s_2}}.
\end{align}
Now let us compute $V_\eta$. From \eqref{eta-sol} it follows that 
$$
    \sprod{\eta}{a(k)\eta} = e^{-i\omega_{0}t}\Big(\sprod{\eta_0}{a(k)\eta_0} + J_t(k)\Big)
$$
and hence that $V_\eta = V_{0,t}+V_{\ph,t}$ where
\begin{eqnarray*}
     V_{0,t}(x) &=& \sqrt{\alpha} \int  \Big(\overline{v(k)}e^{ikx-i\omega_0 t}
     \sprod{\eta_0}{a(k)\eta_0} + {\mathrm h.c.} \Big)dk\\
    V_{\ph,t}(x) &=& \sqrt{\alpha} \int \Big(\overline{v(k)} e^{ikx-i\omega_0 t} J_t(k) + {\mathrm h.c.} \Big)dk\\
&=& -i\alpha \int_0^{t} ds\, e^{i\omega_0(s-t)} \int dk\, |v(k)|^2 \int e^{ik(x-y)}|\ph_s(y)|^2 \,dy + {\mathrm h.c.}\\
&=&  \alpha\int_0^{t} ds \sin(\omega_0(s-t)) \int \frac{|\ph_s(y)|^2}{|x-y|}\,dy.
\end{eqnarray*}

Thus the system \eqref{DF-phi}, \eqref{DF-eta} reduces to the nonlinear Schr\"odinger
equation
\begin{equation}\label{eff-NLS}
     \boxed{i\dot{\ph} = (-\Delta + V_{t})\ph,}
\end{equation}
where $V_t=V_{0,t}+V_{\ph,t}$.  The first part, $V_{0,t}$, is determined by the initial condition $\eta_0$ of the
phonons. It is the potential due to the freely evolving phonons
$e^{-i\omega_0 Nt}\eta_0$ and correspondingly $(\partial_t^2+\omega_0^2)V_{0,t}=0$. The
initial data at $t=0$ are determined by $\eta_0$ through  
$$
     \Big(V_0+i\omega_0^{-1}\partial_t V_0\Big) \Big|_{t=0} = \frac{\alpha}{\pi}\int \frac{dk}{|k|}
     e^{ikx} \sprod{\eta_0}{a(k)\eta_0}.
$$
On the other hand $V_{\ph,t}$, and hence $V_t=V_{0,t}+V_{\ph,t}$, solves the inhomogeneous equation 
\begin{equation}\label{poisson}
     \boxed{(\partial_t^2+\omega_0^2)V_{t} = -\alpha\omega_0 |\cdot|^{-1} * |\ph|^2}
\end{equation}
and $V_{\ph,t}|_{t=0} = 0= \partial_t V_{\ph,t}|_{t=0}$. 

Much of this paper is concerned with a particular solution to \eqref{eff-NLS}, \eqref{poisson} and hence to \eqref{DF-phi}, \eqref{DF-eta}. This particular solution to the system \eqref{eff-NLS}, \eqref{poisson} is of the form 
\begin{equation}\label{pekar-solution}
     \ph_t = e^{-i\lambda t}\ph_0,\qquad V(x) = -\frac{\alpha}{\omega_0}\int\frac{|\ph_0(y)|^2}{|x-y|}\, dy
\end{equation}
where $V$ does not depend on $t$. Then \eqref{eff-NLS} reduces to 
\begin{equation}\label{pekar-equation}
     \lambda \ph_0 = \left(-\Delta - \frac{\alpha}{\omega_0}\int \frac{|\ph_0(y)|^2}{|x-y|}\, dy\right)\ph_0, 
\end{equation}
which is the Euler-Lagrange equation associated with the Pekar functional 
\begin{equation}\label{Pekar}
     \EE_P(\ph)=\int |\nabla\ph(x)|^2\, dx - \frac{\alpha}{2\omega_0} \int\frac{|\ph(x)|^2|\ph(y)|^2}{|x-y|}\, dxdy
\end{equation}
with the constraint $\|\ph\|=1$. It is well-known that this functional has a positive minimizer and infinitely many (spherically symmetric) critical points with negative energy \cite{Lieb1977, Lions1980, Lions1984}. Any critical point $\ph_0$ of \eqref{Pekar} defines, through \eqref{pekar-solution}, a particular solution to \eqref{eff-NLS}, \eqref{poisson}. 

A solution to the system \eqref{DF-phi}, \eqref{DF-eta} that leads to \eqref{pekar-solution} is given by
\begin{equation}\label{solution2LP}
     \ph_t = e^{-i\lambda t}\ph_0,\qquad \eta_t = e^{-i\mu t}\eta_0
\end{equation}
where $\eta_0$ is the coherent state minimizing $\omega_0 N +\sqrt{\alpha}\phi(f)$ and $\mu$ is the minimum. That is,
\begin{equation}\label{eta-0}
     \eta_0 = W\left(\frac{\sqrt\alpha}{\omega_0}f\right)^{*}\Omega,\qquad \mu = -\frac{\alpha}{\omega_0}\|f\|^2
\end{equation}
where $f$ is given by \eqref{def-f} with $\ph=\ph_0$. Indeed, with this choice for $\eta_t$ we obtain from the definitions of $V_{0,t}$ and $V_{\ph,t}$, the expressions
\begin{align*}
      V_{0,t} &= - \frac{\alpha}{\omega_0} \cos(\omega_0 t) \int\frac{|\ph_0(y)|^2}{|x-y|}\, dy \\ 
  V_{\ph,t} &=  \frac{\alpha}{\omega_0}\big(\cos(\omega_0 t)-1\big) \int\frac{|\ph_0(y)|^2}{|x-y|}\, dy
\end{align*}
which add up to $V$ given in \eqref{pekar-solution}. That is, $V_{\eta_0}=V$.

It remains to compute $u_t=a(t)\ph_t\otimes \eta_t$, the solution to the Dirac-Frenkel equation \eqref{DF} associated with \eqref{solution2LP}. In view of $\sqrt{\alpha}\sprod{\ph_0\otimes \eta_0}{\phi(G)(\ph_0\otimes \eta_0)}= -2\frac{\alpha}{\omega_0}\|f\|^2 = 2\mu$, Equations \eqref{solution2LP} and \eqref{DF-a} imply that $a(t) = \exp(2\mu it)$ and hence
\begin{equation}\label{DF-solution}
   u_t = e^{-i E_p t}\ph_0\otimes \eta_0
\end{equation} 
where $E_p := \lambda-\mu = \EE_P(\ph_0)$ is the critical level of the Pekar functional.


\section{Accuracy of the Dirac-Frenkel approximation}
\label{DF-accuracy}

By Proposition \ref{H-rescaling} there is a unitary transformation on the Fock space which turns the Fr\"ohlich Hamiltonian $H_F$ with $\omega_0=1$ into $\alpha^2 \Hres$, where
$\Hres$ is the norm resolvent limit of 
\begin{equation}\label{rescaled-ham}
     \Hrescut = -\Delta+\alpha^{-2}N +\alpha^{-1}\phi(G_\Lambda).
\end{equation}
The factor of $\alpha^2$ in the Schr\"odinger equation $i\partial_t\psi = \alpha^2 \Hres\psi$ is now removed by introducing a microscopic time scale $t_{\rm mic} = \alpha^2 t$, which will be denoted by $t$ in the following. Hence, effectively we replace \eqref{Fro-Ham} by \eqref{rescaled-ham}. The equations of the previous section are converted accordingly by the substitutions
$$
    \sqrt{\alpha}\to\alpha^{-1}, \qquad \omega_0\to \alpha^{-2}
$$
of the parameters $\sqrt{\alpha}, \omega_0$. This means that \eqref{DF-phi}, \eqref{DF-eta} and \eqref{DF-a} become 
\begin{align}
    i\dot{\ph} &= (-\Delta + V_\eta)\ph,\label{DF1}\\ 
   i\dot{\eta} &= (\alpha^{-2} N + \alpha^{-1}\phi(f_t))\eta,\label{DF2} 
\end{align}
and 
\begin{equation}\label{DF3}
      a(t) = \exp \left( i\alpha^{-1}\int_0^t \sprod{\ph_s\otimes \eta_s}{\phi(G)(\ph_s\otimes \eta_s)}\right).
\end{equation}

In the previous section we found particular solutions to \eqref{DF-phi}, \eqref{DF-eta}. They correspond to particular solutions to \eqref{DF1}, \eqref{DF2} given by 
$\ph_t = e^{-i\lambda t}\ph_0$ and $\eta_t = e^{-i\mu t}\eta_0$, where $\ph_0$ is a critical point of the Pekar functional \eqref{Pekar} with $\alpha/\omega_0\to 1$, and 
\begin{equation}\label{def-eta0}
     \eta_0 = W(\alpha f)^{*}\Omega,
\end{equation}
is the minimizer of $\alpha^{-2} N + \alpha^{-1}\phi(f)$, it satisfies $a(k)\eta_0 = -\alpha f(k)\eta_0$. Equation \eqref{def-eta0} agrees with \eqref{eta-0} upon $\sqrt{\alpha},\omega_0\to\alpha^{-1}, \alpha^{-2}$. The function $f$ is defined by
\begin{equation}\label{def-f2}
     f(k) = v(k)\hat\rho(k)
\end{equation}
where $\hat \rho$ is the Fourier transform of $\rho = |\ph_0|^2$. Setting $V = - |\cdot|^{-1}*|\ph_0|^2$ it follows that 
\begin{align}
     (-\Delta +V)\ph_0 &= \lambda \ph_0\label{eq:eval1}\\
    \big( \alpha^{-2} N + \alpha^{-1}\phi(f) \big)\eta_0 &= \mu\eta_0\label{eq:eval2}
\end{align}
where $\mu = -\|f\|^2 = \frac{1}{2}\sprod{\ph_0}{V\ph_0}$.

Suppose now that $\ph_0$ is a minimizer of the Pekar function, rather than just any critical point. This ensures that  $\lambda$ is the lowest eigenvalue of $-\Delta+V$ \cite{AG2014}.
We want to compare
\begin{equation}\label{DF-approx}
    u_t = a(t)\ph_t\otimes \eta_t = e^{-iE_p t}\ph_0\otimes \eta_0,
\end{equation}
to the true evolution $\exp(-i\Hres t)(\ph_0\otimes \eta_0)$. To this end we define
$$
     \Heff := (-\Delta +V)\otimes 1 + 1\otimes (\alpha^{-2}N + \alpha^{-1}\phi(f)) +2\|f\|^2,
$$
which, formally, is nothing but the effective Hamiltonian \eqref{eq:DF-ham} defined in Section \ref{DF-abstract} with $\ph=\ph_0$ and $\eta=\eta_0$ as above. This operator depends on $\ph_0$ only, because $f$ is determined by $\ph_0$. From \eqref{eq:eval1}, \eqref{eq:eval2} and $E_p=\lambda+\|f\|^2$ we see that $\Heff (\ph_0\otimes \eta_0) = E_p (\ph_0\otimes \eta_0)$ and hence 
$$
      u_t = e^{-i\Heff t}(\ph_0\otimes \eta_0).
$$
We thus need to compare the time evolutions of $\ph_0\otimes \eta_0$ generated by $\Hres$ and $\Heff$ respectively. 
By Lemma~\ref{lm:coherent-gauge},
\begin{align*}
     H_\Lambda &:= W(\alpha f) \Hrescut W(\alpha f)^{*} = \Hrescut + V_\Lambda -\alpha^{-1}\phi(f) +\|f\|^2\\
     \tilde H &:= W(\alpha f) \Heff W(\alpha f)^{*} = -\Delta + V +\alpha^{-2}N  +\|f\|^2,
\end{align*}
which implies that
\begin{equation}\label{deltaHcut}
    \delta H_\Lambda := H_\Lambda - \tilde H = \alpha^{-1}(\phi(G_\Lambda)-\phi(f)) + (V_\Lambda - V),
\end{equation}
where $\|V_\Lambda - V\|_{\infty}\to 0$ as $\Lambda\to \infty$, by Lemma~\ref{lm:Vcut-V}. The operator $H = W(\alpha f) \Hres W(\alpha f)^{*} $ is the norm-resolvent limit of $H_\Lambda$ and therefore
\begin{equation}\label{WU}
    W(\alpha f) \left( e^{-i\Hres t} - e^{-i\Heff t} \right)  W(\alpha f)^{*} = \left(e^{-iHt} - e^{-i\tilde Ht} \right).
\end{equation}
With these preparations we may now state and prove our main result:

\begin{theorem}\label{main}
Suppose $\ph_0$ is a minimizer of the Pekar functional, $\eta_0$ is defined by \eqref{def-eta0} and \eqref{def-f2}, and $u_0=\ph_0\otimes\eta_0$. Then there is a constant $C$ such that for all $\alpha\geq 1$ and all $t\in\R$,
$$
     \|e^{-iH_{\alpha}t}u_0 - e^{-i E_p t}u_0\|^2 \leq C \frac{|t|}{\alpha^2}.      
$$
\end{theorem}

\begin{proof}
We shall prove that 
$$
     \|e^{-iH_\Lambda t}\psi_0 - e^{-i\tilde Ht}\psi_0\|^2 \leq C \left(\frac{t}{\alpha^2}\right) + o(1)\qquad (\Lambda\to \infty)
$$
for $\psi_0 = \ph_0\otimes\Omega$ and all $t>0$. Since $H_\Lambda\to H$ as $\Lambda\to\infty$ in the norm-resolvent sense, and hence $e^{-iH_\Lambda t}\psi_0 \to e^{-iHt}\psi_0$, the desired estimate then follows from \eqref{WU} and $u_0 = W(\alpha f)^{*}\psi_0$. In the following $o(1)$ will always stand for a remainder vanishing in the limit $\Lambda\to\infty$.

We begin with
\begin{eqnarray}
  \lefteqn{ \|e^{-iH_\Lambda t}\psi_0 - e^{-i\tilde Ht}\psi_0\|^2}\nonumber \\
   &=& 2\Ima\int_0^t \sprod{e^{-i H_\Lambda s}\psi_0}{\delta H_\Lambda e^{-i\tilde Hs}\psi_0}\, ds\nonumber\\
    &=& 2\Ima\int_0^t \sprod{e^{-i H_\Lambda s}\psi_0-e^{-i\tilde Hs}\psi_0}{\delta H_\Lambda e^{-i\tilde Hs}\psi_0}\, ds\nonumber\\
    &=& 2\Ima\int_0^t \sprod{e^{i\tilde Hs}e^{-i H_\Lambda s}\psi_0-\psi_0}{e^{i(\tilde H - E_p)s}\delta H_\Lambda \psi_0}\, ds.\label{eq1}
\end{eqnarray}
where  $\tilde H\psi_0 = E_p\psi_0$ was used in the last equation. 

Let $P_0:=P_{\ph_0}\otimes P_\Omega$ and $ Q_0 := P_{\ph_0}^{\perp}\otimes P_{\Omega}^{\perp}$. Then 
\begin{equation} \label{Q-delta-H}
       \delta H_\Lambda \psi_0 = P_0^{\perp}  \delta H_\Lambda \psi_0 + o(1)
       = Q_0  \delta H_\Lambda \psi_0 + o(1)
\end{equation}
because, first, 
$$
    \sprod{\psi_0}{\delta H \psi_0} = \sprod{\ph_0}{(V_\Lambda - V)\ph_0} = o(1)
$$
and, second,
\begin{align*}
     (P_{\ph_0}^{\perp}\otimes P_{\Omega})\delta H_\Lambda \psi_0 &= P_{\ph_0}^{\perp}(V_\Lambda - V) \ph_0\otimes \Omega = o(1)\\
     (P_{\ph_0}\otimes P_{\Omega}^{\perp}) \delta H_\Lambda \psi_0 & = \alpha^{-1}\ph_0\otimes (f_\Lambda-f) = o(1),
\end{align*}
where $f_\Lambda(k) := \sprod{\ph_0}{a(G_\Lambda (k))\ph_0} = f(k)\chi_\Lambda(k)$. We have used that, by Lemma~\ref{lm:Vcut-V},  $\|V_\Lambda-V\|_{\infty}\to 0$, $\|f_\Lambda-f\|_2\to 0$, and
that $P_0^{\perp} = (P_{\ph_0}^{\perp}\otimes P_{\Omega}) + (P_{\ph_0}\otimes P_{\Omega}^{\perp}) + (P_{\ph_0}^{\perp}\otimes P_{\Omega}^{\perp})$.

The eigenvalue  $E_p$ of $\tilde H$ associated to $\psi_0=\ph_0\otimes \Omega$ is isolated from the spectrum of  $\tilde H\restricted \ran P_{\ph_0}^{\perp}$ by a gap of the size of the  
gap between the first and the second eigenvalue of $-\Delta+V$. Like $V$ this gap is independent of $\alpha$. Therefore $(\tilde H - E_p)\restricted \ran Q_0$ has a resolvent $\tilde R=(\tilde H - E_p)^{-1}Q_0$ whose norm has a finite bound that is independent of $\alpha$. In the integrand of \eqref{eq1} we now use \eqref{Q-delta-H}, we then write
\begin{align*}
   e^{i\tilde Hs}e^{-i H_\Lambda s}\psi_0 - \psi_0 &= -i\int_0^s e^{i\tilde H\tau} \delta H_\Lambda e^{-i H_\Lambda\tau}\psi_0\, d\tau\\
   e^{i(\tilde H-E_p)s}Q_0 &= -i\frac{d}{ds} e^{i(\tilde H-E_p)s}\tilde{R} Q_0,
\end{align*}
and we integrate by parts to get
\begin{eqnarray}
  \lefteqn{ \|e^{-iH_\Lambda t}\psi_0 - e^{-i\tilde Ht}\psi_0\|^2 + o(1)}\nonumber \\
  &=& 2\Ima \left\langle \int_0^t e^{i\tilde Hs}\delta H_\Lambda e^{-i H_\Lambda s}\psi_0\, ds, e^{i(\tilde H-E_p)t} \tilde{R} Q_0
    \delta H_\Lambda \psi_0\right\rangle\nonumber \\
  & & -2\Ima \int_0^t \sprod{e^{i\tilde Hs}\delta H_\Lambda e^{-iH_\Lambda s}\psi_0}{e^{i(\tilde H-E_p)s}\tilde{R}Q_0 \delta H_\Lambda \psi_0}\, ds\nonumber\\
  & =& A_\Lambda(t) - B_\Lambda (t).\label{def-AB}
\end{eqnarray}
Both, $A_\Lambda (t)$ and $B_\Lambda(t)$ contain two factors of $\delta H_\Lambda$ and thus two factors of $\alpha^{-1}$. We therefore expect that both $A_\Lambda (t)$ and $B_\Lambda(t)$ satisfy the desired estimate. To prove this we first note that 
\begin{equation}\label{deltaHpsi}
     Q_0 \delta H_\Lambda \psi_0 = \alpha^{-1} Q_0 a^{*}(G_\Lambda)\psi_0,
\end{equation}
which follows from $P_\Omega^{\perp}(V_\Lambda-V)\psi_0=0$, $P_{\ph_0}^{\perp}\phi(f)\psi_0=0$, and $a(G_\Lambda)\psi_0=0$. Let $\psi_s := e^{-iH_\Lambda s}\psi_0$. By \eqref{deltaHcut} and \eqref{deltaHpsi}, 
\begin{align*}
    B_\Lambda(t) = &\alpha^{-2} 2\Ima\int_0^{t} \sprod{(\phi(G_\Lambda)-\phi(f)) \psi_s}{\tilde R Q_0 a^{*}(G_\Lambda)\psi_0}e^{-iE_ps}\, ds\\
         &+ \alpha^{-1} 2\Ima\int_0^{t} \sprod{(V_\Lambda - V) \psi_s}{\tilde R Q_0 a^{*}(G_\Lambda)\psi_0}e^{-iE_ps}\, ds.
\end{align*}
The second term vanishes in the limit $\Lambda\to \infty$ thanks to Lemma~\ref{lm:Lieb} and Lemma~\ref{lm:Vcut-V}.
From $\phi(G_\Lambda) = a(G_\Lambda)+a^{*}(G_\Lambda)$ we get two contributions to the integrand of the first term. For the contribution of $a^{*}(G_\Lambda)$ we obtain, using $\|\psi_s\|=\|\psi_0\|=1$, Lemma~\ref{lm:FSch} and Lemma~\ref{lm:Lieb},
\begin{eqnarray*}
\left| \sprod{  a^{*}(G_\Lambda)\psi_s}{\tilde R Q_0 a^{*}(G_\Lambda)\psi_0} \right|
   &\leq & \| a(G_\Lambda) \tilde R Q_0 a^{*}(G_\Lambda)\psi_0 \|\\
   &\leq & C_v \| (1+p^2)^{1/2}N^{1/2} \tilde R Q_0 a^{*}(G_\Lambda)\psi_0 \|\\
   &\leq & C_v \| (1+p^2)^{1/2}\tilde R^{1/2} Q_0\| \|  \tilde R^{1/2} Q_0 a^{*}(G_\Lambda)\psi_0 \|\\
   &\leq & C \|\psi_0\|.
\end{eqnarray*}
Here and below it is used that $ \tilde R Q_0 a^{*}(G_\Lambda)\psi_0 $ is a one-phonon state. For the contribution of $a(G_\Lambda)$ to the integrand we obtain
by using Lemma~\ref{lm:FSch} and Lemma~\ref{lm:Lieb},
\begin{eqnarray*}
  \lefteqn{\left| \sprod{  a(G_\Lambda)\psi_s}{\tilde R Q_0 a^{*}(G_\Lambda)\psi_0} \right|} \\
&\leq & \|(N+1)^{-1/2} a(G_\Lambda) \psi_s\| \| (N+1)^{1/2} \tilde R Q_0 a^{*}(G_\Lambda)\psi_0 \|\\
&\leq & C_v \sqrt{2} \|(1+p^2)^{1/2} \psi_s\| \| \tilde R Q_0 a^{*}(G_\Lambda)\psi_0 \|\\
&\leq & C \|(1+p^2)^{1/2} \psi_s\|, 
\end{eqnarray*}
where, by Lemma~\ref{H-bounds}, 
\begin{align*}
  \|(1+p^2)^{1/2} \psi_s\|^2 &\leq 2\sprod{\psi_s}{(H_\Lambda+C)\psi_s}\\ 
  &=2\sprod{\psi_0}{(H_\Lambda+C)\psi_0} = \sprod{\ph_0}{(-\Delta + V_\Lambda +C)\ph_0}
\end{align*}
which is bounded uniformly in $\Lambda>0$. Since the contribution of $\phi(f)$ to the integrand of  $B_\Lambda(t)$ can be estimated in the same way, we conclude that $|B_\Lambda(t)| \leq  C t \alpha^{-2}$. In a very similar way one shows that $|A_\Lambda(t)| \leq C t \alpha^{-2}$.
\end{proof}


\section{More general initial states}
\label{DF-general}

We now use the method developed in the previous sections for comparing the Dirac-Frenkel evolution with the true quantum evolution of more general initial states than those considered before. These initial states are still of the product form $u_0 = \ph_0\otimes \eta_0$ with $\ph_0\in L^2(\R^3)$ and a coherent state $\eta_0\in\FF$ but in addition we only require that 
\begin{align*}
     \ph_0 \in H^2(\R^3)\quad \text{and}\quad \eta_0 = W(\alpha g)^{*}\Omega,\qquad g\in L^2(\R^3,(1+k^2)dk).
\end{align*} 
For such initial data the system \eqref{DF1},  \eqref{DF2}  has a unique global solution $\ph_t$, $\eta_t$, $t\in\R$, with $\ph_t\in H^2(\R^3)$ depending continuously on time \cite{FG2015,GSS2016}.
This well-posedness result is only available in $d=3$ space dimensions so far. Let  $a(t)$ be defined by \eqref{DF3} and let 
$$
       u_t:= a(t)\ph_t\otimes \eta_t, \qquad t \in \R.
$$ 
Then the following theorem holds true:

\begin{theorem}\label{main-2}
There exists a constant $C\in\R$ such that for all $\alpha\geq 1$ and all $t\in\R$,
$$
     \|e^{-iHt}u_0 - u_t\|^2 \leq C \left | \frac{t}{\alpha}\right|.
$$
\end{theorem}

\noindent
\emph{Remark.} The constant $C$ does not depend on $g$ and hence the $\alpha$-dependence of $\eta_0$ is actually immaterial.

\begin{proof}
As in the proof of Theorem~\ref{main} we should estimate $ \|e^{-iH_{\Lambda}t}u_0 - u_t\|$ for $\Lambda<\infty$ and then let $\Lambda\to\infty$. This could easily be done, but it would obscure the structure and simplicity of the proof. We therefore proceed formally setting $\Lambda=\infty$. From the definition of $u_t$ it follows that $i\dot u_s = H_s u_s$ with 
\begin{equation}\label{def-Hs}
    H_s := -\Delta +V_{\eta_s} +\alpha^{-2}N +\alpha^{-1}\phi(f_s) - \alpha^{-1}\sprod{\ph_s\otimes \eta_s}{\phi(G)\ph_s\otimes \eta_s}
\end{equation}
where $f_s(k) = \sprod{\ph_s}{G(k)\ph_s}=v(k)\hat\rho_s(k)$ and $\rho_s=|\ph_s|^2$. We recall from Section~\ref{DF-polaron} that the equation \eqref{DF2} for $\eta_s$ can be solved explicitly in terms of $\ph_s$ with the result 
\begin{equation}\label{evolve-eta}
     \eta_s = e^{-iF(s)} W(h_s)\Omega
\end{equation}
with some real-valued function $F(s)$ and $h_s\in L^2(\R^d)$, which will drop out in the end. It is important however that the evolved state is again a coherent state. 
Writing $V_{\eta_s}(x)$ as $\sprod{\eta_s}{\phi(G_x)\eta_s}$ and $\sprod{\ph_s\otimes \eta_s}{\phi(G)\ph_s\otimes \eta_s} $ as $\sprod{\eta_s}{\phi(f_s)\eta_s}$ in \eqref{def-Hs} we obtain
\begin{equation}\label{diff-ham}
     H-H_s = \frac{1}{\alpha}\left(\phi(G) - \phi(f_s)\right) -  \frac{1}{\alpha} \sprod{\eta_s}{\phi(G)\eta_s} - \frac{1}{\alpha} \sprod{\eta_s}{\phi(f_s)\eta_s}.
\end{equation}
where the index $x$ has been dropped again. From \eqref{evolve-eta} we see that 
\begin{align*}
     W(h_s)^{*}\phi(G)W(h_s)&= \phi(G) + \sprod{\eta_s}{\phi(G)\eta_s}\\
     W(h_s)^{*}\phi(f_s)W(h_s) &= \phi(G) + \sprod{\eta_s}{\phi(f_s)\eta_s}
\end{align*}
which, in view of \eqref{diff-ham}, leads to 
$$
      W(h_s)^{*}(H-H_s )W(h_s) = \frac{1}{\alpha} (\phi(G) - \phi(f_s)).
$$

Writing $u_s = W(h_s)(\ph_s\otimes \Omega)e^{-iF(s)}$ we conclude that 
\begin{align*}
   \|e^{-iHt}u_0 - u_t\|^2 &= 2\Ima \int_0^t \sprod{e^{-iHs}u_0}{(H-H_s)u_s}\, ds\\
   &= 2\Ima \int_0^t \sprod{W(h_s)^{*} e^{-iHs}u_0}{W(h_s)^{*}(H-H_s)W(h_s) \ph_s\otimes \Omega}e^{-iF(s)}\, ds\\
   &= \frac{2}{\alpha}\Ima \int_0^t \sprod{W(h_s)^{*} e^{-iHs}u_0}{(\phi(G) - \phi(f_s)) \ph_s\otimes \Omega}e^{-iF(s)}\, ds,
\end{align*}
where the annihilation operators $a(G)$ and $a(f_s)$ give no contribution to the integrand. For the contribution due to $a^{*}(f_s)$ we have the bound
$$
     \frac{2t}{\alpha} \sup_s\|f_s\| \|\ph_s\|
$$
where $\|\ph_s\| = \|\ph_0\|=1$ and
\begin{align*}
     \|f_s\|^2 &= \int |v(k)|^2|\hat\rho_s(k)|^2\, dk\\
     &= \frac{1}{2} \int\frac{\rho_s(x)\rho_s(y)}{|x-y|}\, dxdy \leq C \|\nabla \ph_s\|^2
\end{align*}
which is bounded uniformly in $s$, as we show below. The contribution of  $a^{*}(G)$ to the integrand is bounded above by
\begin{eqnarray*}
   \lefteqn{| \sprod{W(h_s)e^{-iHs}u_0}{a^{*}(G) \ph_s\otimes \Omega}|} \\ &\leq \|(1+p^2)^{1/2} e^{-iHs}u_0\| \cdot \|(1+p^2)^{-1/2}a^{*}(G)\ph_s\otimes \Omega\|=: A_s B_s.
\end{eqnarray*}
With the help of Lemma \ref{H-bounds} we see that 
$$
       A_s^2 \leq 2(\|\nabla\ph_0\|^2 +C)
$$
and by part (b) of Lemma \ref{lm:Lieb},
$$
     B_s \leq C (\|\nabla\ph_s\| + \|\ph_s\| ).
$$

It remains to prove that $\sup_s \|\nabla \ph_s\| <\infty$. This follows from 
$$
    \sprod{u_s}{H u_s} \geq \|\nabla\ph_s\|^2 -  \frac{1}{2} \int\frac{\rho_s(x)\rho_s(y)}{|x-y|}\, dxdy \geq \frac{1}{2}\|\nabla\ph_s\|^2 - C
$$
and the general fact that Dirac-Frenkel evolution conserves the energy, that is $ \sprod{u_s}{Hu_s} =  \sprod{u_0}{Hu_0}$, see Section~\ref{DF-abstract}.
\end{proof}


\section{Extension to $N$ polarons}
\label{sec:N-polaron}

The analysis of Section~\ref{DF-accuracy} can be carried out equally for $N$-polaron systems, provided the parameters are in a range where the Pekar functional for $N$ polarons, the so called Pekar-Tomasevich functional, has a minimizer. We now elaborate on these remarks in some mathematical detail, but we shall mostly suppress the UV-cutoff for notational simplicity. As in the previous section we here depend on results that are only available in $d=3$ space dimensions so far.

The Fr\"ohlich Hamiltonian for $N$ polarons in $\R^3$ is of the form 
$$
   H_N= \sum_{j=1}^{N}\left(-\Delta_{x_j} +\alpha^{-1}\phi(G_j)\right)+ \sum_{i<j}\frac{U}{|x_i-x_j|}+\alpha^{-2}\Nph
$$
and it acts on the Hilbert space $\HH=\HH_{part}\otimes \FF$ where $\HH_{part}$ may be the symmetric, the antisymmetric, of the full $N$-fold product of $N$ copies of $L^2(\R^3)$. $G_j(k)$ denotes the bounded operator on $\HH_{part}$ defined by multiplication with $e^{-ikx_j}v(k)$ where $x_j\in \R^3$ denotes the position of the $j$th particle and $v(k)=1/(2\pi|k|)$. The operator $H_N$ is of the familiar form $H_N= H_0+W$ where $H_0 = H_1\otimes 1 + 1\otimes H_2$ is the sum of operators on $\HH_{part}$ and $\FF$ respectively. The Dirac-Frenkel equations  \eqref{dot-ph3} and \eqref{dot-eta3} now take the form
\begin{align}
    i\dot{\ph} &= \Big(\sum_{j=1}^N(-\Delta_{x_j}) + \sum_{i<j}\frac{U}{|x_i-x_j|} + V_\eta(x_1,\ldots,x_N)\Big)\ph,\label{DFN-phi}\\
   i\dot{\eta} &= (\alpha^{-2} N + \alpha^{-1}\phi(f_t))\eta,\label{DFN-eta} 
\end{align}
where
$$
     V_\eta(x_1\ldots,x_N) := \alpha^{-1}\sum_{j=1}^N\int \Big(\overline{v(k)}e^{ikx_j}\sprod{\eta}{a(k)\eta} + {\rm h.c}\Big)\, dk
$$
and $f(k) = v(k)\hat\rho(k)$. The function $\hat\rho$ is the Fourier transform of the particle density defined by 
$$
   \rho(x):= \sum_{j=1}^N\int |\ph (\ldots x_{j-1}, x, x_{j+1}\ldots)|^2\, dx_1,\ldots \widehat{dx_j}\ldots dx_N
$$
where $x$ takes the place of $x_j$ in the argument of $\ph$ and integration w.r.t $x_j$ is omitted. The equations \eqref{DFN-phi} and \eqref{DFN-eta} are to be compared with \eqref{DF-phi} and \eqref{DF-eta}. Following \eqref{solution2LP} we now make the Ansatz
\begin{equation}\label{solution2LPN}
     \ph_t = e^{-i\lambda t}\ph_0,\qquad \eta_t = e^{-i\mu t}\eta_0
\end{equation}
with $\lambda,\mu\in \R$, $\ph_0\in \HH_{part}$ and $\eta_0\in \FF$ to be determined. The potential $V_\eta$ and the function $f$ are now independent of time. Given \eqref{solution2LPN}, the equation \eqref{DFN-eta} is solved if 
$$
      \eta_0 = W(\alpha f)^{*}\Omega
$$
and $\mu = - \|f\|^2$. For the potential $V_\eta$ this means that 
$$
          V_\eta(x_1\ldots,x_N) = -\sum_{j=1}^N \int \frac{\rho(y)}{|x_j-y|}\, dy
$$
and hence the equation for $\ph_0$ and $\lambda$ becomes the Euler-Lagrange equation for the \emph{Pekar-Tomasevich} functional
\begin{align*}
    \EE_N(\ph) & = \Big\langle \ph,\Big(\sum_{j=1}^{N}(-\Delta_{x_j}) +
        \sum_{i<j}\frac{U}{|x_i-x_j|}\Big)\ph\Big\rangle
      -\frac{1}{2}\int\frac{\rho(x)\rho(y)}{|x-y|}\,dxdy
\end{align*} 
with the constraint $\|\ph\|=1$. From \cite{Lewin2011, AG2014} we know that $\EE_N$ has a minimizer $\ph_0$, provided that $U<1+\eps_N$ where $\eps_N>0$. The minimum $E_N = \EE_N(\ph_0)$ and the Lagrange multiplier $\lambda$ are related by $E_N=\lambda-\mu$. In view of the phase 
$$
    a(t) = \exp(i \int_0^t \sprod{\ph_0\otimes \eta_0}{W\ph_0\otimes \eta_0}\,ds) = \exp(2i\mu t)
$$
and $\lambda+\mu-2\mu = E_N$ it is 
$$
       u_t = a(t)\ph_t\otimes \eta_t = e^{-iE_Nt}(\ph_0\otimes \eta_0)
$$
that we want to compare to the true time evolution generated by $H_N$. A copy of the proof of Theorem~\ref{main} shows that:

 \begin{theorem}\label{main-N}
 Let $N\geq 2$ and $u_0 = \ph_0\otimes\eta_0\in \HH$ with $\ph_0$, $\eta_0$ as defined above. Let $E_N=\EE_N(\ph_0)$. Then there exists a constant $C_N$ such that, for all $t\in\R$ and $\alpha\geq 1$,
$$
     \| e^{-iH_Nt}u_0 - e^{-i E_N t} u_0\|^2 \leq C_N \frac{|t|}{\alpha^2}.
$$
\end{theorem}

\appendix
\section{Technical tools}

\begin{lemma}\label{lm:Vcut-V}
Let $\ph_0\in L^2(\R^d)$, $d\in\{2,3\}$, be a minimizer of the Pekar functional \eqref{Pekar} (with $\alpha/\omega_0 =1$), let $V = - |\cdot|^{-1}*|\ph_0|^2$ and let $V_\Lambda$ be defined by $\hat{V}_\Lambda = \hat{V}\chi_{\Lambda}$. Then 
$$
      \|V_\Lambda - V \|_{\infty} \to 0,\qquad (\Lambda\to\infty).
$$
\end{lemma}

\begin{proof}
It is well known that $H^{s}(\R^d)$ is an algebra for $s>d/2$. Since $\ph_0\in H^{2}(\R^d)$ with $2>d/2$ it follows that $\rho=|\ph_0|^2\in H^2(\R^d)$, and hence 
\begin{align*}
     |(V_\Lambda - V)(x)| &= C_d\left|\int_{|k|>\Lambda} e^{ikx}\frac{1}{|k|^{d-1}}\hat{\rho}(k)\, dk \right|\\
     &\leq C_d\left(\int_{|k|>\Lambda} \frac{1}{|k|^{2d}} dk\right)^{1/2} \left(\int |k|^2 |\hat\rho(k)|^2\, dk \right)^{1/2} \to 0\qquad (\Lambda\to\infty).
\end{align*}
\end{proof}

\begin{lemma}[Frank, Schlein]\label{lm:FSch}
Let $G_\Lambda : L^2(\R^d) \to  L^2(\R^d)\otimes  L^2(\R^d)$ be defined by $(G_\Lambda\ph)(x,k)=\ph(x)e^{-ikx}v(k)\chi_\Lambda(k)$, with $v\in L^2_{\rm loc}(\R^d)$ such that 
$$
    C_v :=\sup_{q\in\R^d} \int\frac{|v(k)|^2}{1+(q-k)^2}\,dk <\infty.
$$
Then, for all $\Lambda>0$,
\begin{eqnarray*}
   \|a(G_\Lambda)\psi\| &\leq & C_v \|(1+p^2)^{1/2}N^{1/2}\psi\|\\
   \|(N+1)^{-1/2}a(G_\Lambda)\psi\| &\leq &C_v \|(1+p^2)^{1/2}\psi\|
\end{eqnarray*}
\end{lemma}

\begin{proof}
The second inequality follows from $(N+1)^{-1/2}a(G_\Lambda)\psi = a(G_\Lambda)N^{-1/2}\psi$ and from the first one. For the proof of the first inequality we refer to \cite{FrankSchlein2014, GriWue2016}. 
\end{proof}

\begin{lemma}\label{lm:Lieb}
Let $\ph_0\in L^2(\R^d)$ be a minimizer of the Pekar functional, let $Q_0 = P_{\ph_0}^{\perp}\otimes P_\Omega^{\perp}$ and let $\tilde R$ be the resolvent of $\tilde H-E_p$ on $\ran Q_0$. Then
\begin{align*}
     (a)\quad &\sup_{\alpha>0} \|(1+p^2)^{1/2} \tilde R^{1/2} Q_0\| <\infty\\
     (b)\quad &\sup_{\alpha,\Lambda>0} \|\tilde R^{1/2} Q_0 a^{*}(G_\Lambda)(\ph_0\otimes \Omega)\| <\infty,
\end{align*}
where $G_\Lambda$ and $\tilde{H}$ are defined in Sections~\ref{DF-polaron} and \ref{DF-accuracy}.
\end{lemma}

\begin{proof}
For any normalized vector $\psi\in D(-\Delta+N)$ we obtain, using $p^2\leq \tilde H+\|V\|_{\infty}$, that
\begin{align*}
     \|(1+p^2)^{1/2} \tilde R^{1/2} Q_0\psi\|^2 & \leq \sprod{Q_0\psi}{\tilde R^{1/2}(\tilde H+\|V\|_{\infty} + 1)\tilde R^{1/2}Q_0\psi}\\
     & = \sprod{Q_0\psi}{(1 + (E_p+\|V\|_{\infty} + 1)\tilde R)Q_0\psi}\\
     & \leq 1 + (E_p+\|V\|_{\infty} + 1)\|\tilde R Q_0\|.
\end{align*}
Since $E_p=\lambda+\| f \|^2$ we have $\tilde H- E_p \geq -\Delta +V -\lambda$ and 
$\|\tilde R Q_0\| \leq \|(-\Delta +V -\lambda)^{-1}P_{\ph_0}^{\perp}\|$, which is finite and independent of $\alpha$.
This completes the proof of (a). 

To prove (b) we write
\begin{align*}
     G_\Lambda(k) &= G_\Lambda(k)\chi_{|k|\leq 1} - \frac{1}{|k|^2}[p\cdot k,
     G_\Lambda (k)]\chi_{|k|> 1}  \\
       &=  G_\Lambda(k)\chi_{|k|\leq 1} - \frac{1}{|k|^2} (p\cdot k G_\Lambda(k) - G_\Lambda(k)
       k\cdot p )\chi_{|k|>1}
\end{align*}
All three terms give rise to uniformly bounded contributions to $\|\tilde R^{1/2} Q_0 a^{*}(G_\Lambda)(\ph_0\otimes \Omega)\|$ because 
$$
    \int_{|k|>1} \frac{|v(k)|^2}{|k|^2}\, dk  <\infty
$$
and because, similar to (a), $\sup_{\alpha>0} \|\tilde R^{1/2} Q_0 p\| <\infty$.
\end{proof}

\begin{lemma}\label{lm:coherent-gauge}
Let $f(k)=v(k)\hat\rho(k)$, $V=-|\cdot|^{-1}*\rho$ and let $V_\Lambda$ be defined by $\hat V_\Lambda(k)=\chi_\Lambda(k)\hat V(k)$. Then 
\begin{align*}
       W(\alpha f) \Hrescut W(\alpha f)^{*} &= H_\Lambda +V_\Lambda -\alpha^{-1}\phi(f) +\|f\|^2,\\
        W(\alpha f) \Heff W(\alpha f)^{*} &= -\Delta + V +\alpha^{-2}N +\|f\|^2,
\end{align*}
and 
$$
    W(\alpha f) (\Hrescut - \Heff) W(\alpha f)^{*} = \alpha^{-1}\big(\phi(G_\Lambda)-\phi(f)\big) + (V_\Lambda - V). 
$$
\end{lemma}

\begin{proof}
From
\begin{align*}
    W(\alpha f) a(k) W(\alpha f)^{*} &= a(k)-\alpha f(k)\\
     W(\alpha f) a^*(k) W(\alpha f)^{*} &= a(k)-\alpha \overline{f(k)}
\end{align*}
it follows that 
\begin{align}
    W(\alpha f) \alpha^{-2}N W(\alpha f)^{*} &=  \alpha^{-2}N -  \alpha^{-1}\phi(f) + \|f\|^2\label{cg1}\\
     W(\alpha f) \alpha^{-1}\phi(f) W(\alpha f)^{*} &=  \alpha^{-1}\phi(f) -2\|f\|^2 \label{cg2}
\end{align}
and
\begin{equation}\label{cg3}
   W(\alpha f) \alpha^{-1}\phi(G_{\Lambda,x}) W(\alpha f)^{*} =  \alpha^{-1}\phi(G_{\Lambda,x}) +V_\Lambda(x),
\end{equation}
where we used that $(2\pi)^d |v(k)|^2$ is the Fourier transform of $|\cdot |^{-1}/2$. The lemma follows from \eqref{cg1}, \eqref{cg2} and \eqref{cg3}. 
\end{proof}

\begin{lemma}\label{H-bounds}
For every $\eps\in (0,1)$ there exists a constant $C_\eps$ such that for all $\Lambda>0$, $\alpha\geq 1$,
$$
   (1-\eps)\left(-\Delta + \alpha^{-2}N\right) - C_\eps\ \leq\  H_{\alpha,\Lambda} \ \leq\  (1+ \eps)\left(-\Delta + \alpha^{-2}N\right) + C_\eps.
$$
\end{lemma}

\begin{proof}
From Lemma 2.1 of \cite{GriWue2016} we know that, for all  $\Lambda\geq \Lambda_0>0$,
\begin{eqnarray}\nonumber
   \pm\alpha^{-1}\left(\phi(G_\Lambda) - \phi(G_{\Lambda_0})\right) &\leq & 2\left(\int_{|k_0|\geq \Lambda_0} \frac{|v(k)|^2}{k^2}\right)^{1/2}(-\Delta + \alpha^{-2}N +\alpha^{-2})\\ 
    &\leq &\eps (-\Delta + \alpha^{-2}N +\alpha^{-2})\label{Hb1}
\end{eqnarray} 
for $\Lambda_0$ large enough. On the other hand, for all $\Lambda\leq \Lambda_0$,
\begin{equation}\label{Hb2}
 \pm\alpha^{-1}\phi(G_{\Lambda}) \leq \frac{\eps}{\alpha^2}N + \frac{1}{\eps} \int_{|k|\leq \Lambda_0} |v(k)|^2\, dk.
\end{equation}
From \eqref{Hb1} and \eqref{Hb2} combined it follows that, for all $\Lambda>0$
$$
     \pm\alpha^{-1}\phi(G_{\Lambda}) \leq \eps (-\Delta + 2\alpha^{-2}N +\alpha^{-2}) + \frac{1}{\eps} \int_{|k|\leq \Lambda_0} |v(k)|^2\, dk,
$$
which proves the desired estimates.
\end{proof}

\section{Rescaling of the Fr\"ohlich Hamiltonian}

let $H_{F}$ and $\Hres$ be the self-adjoint Hamiltonians defined in terms
of the norm resolvent limits of the cutoff Hamiltonians
\begin{align*}
     H_{F,\Lambda} &= -\Delta +N + \sqrt{\alpha}\phi(G_{\Lambda})\\
     H_{\alpha,\Lambda} &= -\Delta+\alpha^{-2}N + \alpha^{-1}\phi(G_{\Lambda}).
\end{align*}
where $G_{\Lambda,x}(k) = c_d |k|^{-(d-1)/2}\chi_\Lambda(k)e^{-ikx}$.

\begin{prop}\label{H-rescaling}
There exists a unitary transformation $U$ depending on $\alpha>0$ such
that $U^{*} H_{F}U = \alpha^2 \Hres$.
\end{prop}

\begin{proof}
Let $U_\alpha$ in $L^2(\R^d)$ be defined by $(U_\alpha\psi)(x) =
\alpha^{d/2}\psi(\alpha x)$ and let $U=U_\alpha\otimes
\Gamma(U_\alpha^{*})$. We claim that
\begin{equation}\label{scaling}
    U^{*} H_{F,\Lambda}U = \alpha^2 (-\Delta) + N +\alpha \phi(G_{\Lambda/\alpha}) =
    \alpha^2 H_{\alpha,\Lambda/\alpha}.
\end{equation}
From \eqref{scaling} the assertion follows in the limit
$\Lambda\to\infty$. To prove \eqref{scaling} we use that 
$$
   \big((U_\alpha^{*}\otimes 1)\phi(G_\Lambda)\psi\big)(x) =
   \alpha^{-d/2}\phi(G_{\Lambda,x/\alpha})\psi(x/\alpha) =  \phi(G_{\Lambda,x/\alpha}) (U_\alpha^{*}\otimes 1)\psi(x) 
$$
and 
$$
     \Gamma(U_\alpha) \phi(G_{\Lambda,x/\alpha})\Gamma(U_\alpha^{*}) =
    \phi(U_\alpha G_{\Lambda,x/\alpha}) =  \sqrt{\alpha}\phi(G_{\Lambda/\alpha,x}).
$$
In the last equation we used the scaling properties of $G_{\Lambda,x}$, 
which imply that 
$$
       U_\alpha G_{\Lambda,x/\alpha}(k) =
       \alpha^{d/2}G_{\Lambda,x/\alpha}(\alpha k)= \sqrt{\alpha}G_{\Lambda/\alpha,x}(k).
$$
\end{proof}

\noindent
\textbf{Acknowledgement:} The author appreciates the many discussions with Joachim Kerner and Andreas W\"unsch  at an early stage of this work and he thanks 
Nicolas Rougerie for asking about the possibility to extend the main result to $N$-polaron systems. Kerner and W\"unsch were supported by the \emph{Deutsche Forschungsgemeinschaft} through the Graduiertenkolleg 1838, "Spectral Theory and Dynamics of Quantum System".



\end{document}